\documentclass[runningheads]{llncs}

\usepackage[T1]{fontenc}

\usepackage{graphicx}

\usepackage{xspace}
\usepackage{csquotes}
\usepackage{amsmath}
\usepackage{amssymb}
\usepackage{hyperref}
\usepackage{cleveref}
\usepackage{subcaption}

\usepackage{tikz}
\usetikzlibrary{arrows,shapes}
\usetikzlibrary{positioning}

\usepackage{enumitem}

\usepackage{fullmacros}
\renewcommand{\initial}[1]{\mathbf{i}(#1)}
\renewcommand{\terminal}[1]{\mathbf{t}(#1)}

\begin{document}

\title{Mathematical Informatics: Algorithms\thanks{This work was partially funded by the ANR ANR-22-CE48-0003-01 project DySCo.}}

\author{Thomas Seiller\inst{1}\orcidID{0000-0001-6313-0898}}

\authorrunning{T. Seiller}

\institute{CNRS, JFLI -- IRL 3527, Tokyo, Japan\\
\email{thomas.seiller@cnrs.fr}\\
\url{http://www.seiller.org}}

\maketitle              

\begin{abstract}
This work continues the development of an intensional approach to computability initiated in previous work \cite{seiller-HdR,seiller-MI1}, in which programs and computations, rather than functions, constitute the primary objects of study. In this setting, models of computation are described as monoid actions on a configuration space, and programs as dynamical systems constrained by this action.

Within this framework, we introduce a formal notion of algorithm as a finite directed graph whose edges are labelled by partial maps over an abstract data structure. This definition separates control from data, representing the former as a graph and the latter as an algebra of operations. 
We then define what it means for a program, in a given model of computation, to implement such an algorithm, by requiring a correspondence between computational steps and labelled transitions that preserves the induced transformations on representations of data. This yields a precise notion of implementation and situates algorithms as abstract partial specifications of computational behaviour.

\end{abstract}

\section{Introduction}

The terminology “algorithm” predates the first computers, and does not appear in the earliest works on computability. It originally referred to methods of resolution in mathematics, as reflected for instance in the term algoriste \cite{LamasseAlgoriste}. Its systematic use within computer science appears later, notably within the Russian school \cite{MarkovAlgo,KolmogorovUspensky}, even though the term also already occurs in Church’s 1936 paper \cite{ChurchAlgo}. It was later widely adopted and became central to the discipline, to the point of naming a subfield, algorithmics. More recently, the term has acquired a broader meaning in common usage, where it is used to refer to complex computational systems, in particular machine learning models and platform-based decision systems.

These different usages have been analysed by Airoldi \cite{Airoldi}, who distinguishes between an analog era, a digital era, and a platform era. While this periodisation is insightful, it should not obscure the fact that these meanings coexist today. The term “algorithm” is thus used to designate heterogeneous objects, depending on context. This raises a natural question: do these usages refer to a common underlying notion?

The standpoint adopted in this work is that the mathematical (analog) and computer science (digital) notions coincide at an abstract level, and that they admit a common formalisation. The apparent differences arise from constraints related to implementation and physical realisation, rather than from the notion of algorithm itself. In particular, restricting to digital representations does not adequately capture the diversity of computational paradigms, including analog, quantum, or biological ones. This suggests that the notion of algorithm should be studied independently of any specific choice of model of computation, neither any choice of specific representation of data.

I therefore focus on the problem of formalising the notion of algorithm itself, independently of its implementations. This requires identifying a notion which is sufficiently abstract to encompass the various usages discussed above, while remaining mathematically tractable. Founded upon the intensional approach to computability developed in \cite{seiller-HdR,seiller-MI1}, I propose a definition of algorithms as finite graphs labelled by operations over an abstract data structure. This definition separates control and data, and integrates naturally with the notions of program and computation introduced in the intensional framework. In particular, it allows for a precise formulation of the notion of implementation, relating algorithms to concrete programs across different models of computation. 

By grounding programs in abstract models of computation defined as monoid actions, and algorithms in labelled graphs over abstract data structures, the framework cleanly separates the two notions while making their relationship, implementation via glueing, formally explicit and mathematically tractable. The present proposal therefore satisfies the following four desiderata, which we believe an adequate formalisation of algorithms should satisfy \cite{GoA}:
\begin{enumerate}[noitemsep,nolistsep]
\item distinguish algorithms from the programs that implement them;
\item remain independent of any particular model of computation;
\item provide a mathematically precise notion of implementation;
\item allow meaningful comparison between algorithms and implementations.
\end{enumerate} 

\paragraph{Comparison to other approaches.}

The question of formalising algorithms has been addressed by several authors. Notably, Moschovakis \cite{MoschovakisAlgo,MoschovakisFoundations} and Gurevich \cite{GurevichAlgo,BlassGurevich} have proposed influential frameworks. 

The framework proposed here offers several theoretical advantages over existing approaches to formalising the notion of algorithm. Gurevich's approach shares with the current proposal the definition of algorithms as partial specifications. However, despite their generality as a model of computation, abstract state machines do not explicitly distinguish between algorithms and programs, nor do they provide a formal notion of implementation. As such, they do not directly support comparison between different implementations, programs, or algorithms. 

Moschovakis' recursors fare better in capturing what we intuitively think of as algorithms, but they are limited to recursive functions and the framework does not satisfactorily addresses the program-level. Indeed, the implementation notion via iterators does not provide explicit treatment of representations and fails to address the diversity of choices of models of computation. 

\subsection{Contents}

The proposed mathematical definition of algorithms will be exposed in three steps. This decomposition is meaningful in itself, and provide a pedagogical approach. The first notion is that of a \emph{syntactic algorithm}. A syntactic algorithm describes the structure of a computation: the collection of elementary actions involved and the way in which they are composed. In particular, it does not prescribe any semantics for these actions, i.e. it does not specify the transformations they induce on data.

As such, syntactic algorithms capture a commonly accepted intuition: an algorithm is given by a finite description of successive operations organised by a control structure. However, this description leaves implicit a number of assumptions regarding the meaning of those operations. Making these assumptions explicit is a central issue in any attempt to formalise the notion of algorithm.

Consider for instance the algorithm shown in Figure \ref{gcd2}. It can be found in this shape in numerous textbooks on algorithmics. It is unclear, however, how the operations indicated in the text should be understood. What is the meaning of the operations ``$-$'' and ``$\bmod$''? What is the domain of the variables $x$ and $y$? In standard presentations, one assumes that $x$ and $y$ range over the integers and that these symbols denote the usual arithmetic operations. But nothing in the description itself enforces this interpretation. In fact, much information is usually left implicit: whether the inputs are natural numbers, or polynomials; whether the relation $\neq$ or the operation $\bmod$ are the usual, intuitive ones, or not. 

This observation leads to the idea that the notion of syntactic algorithm should be complemented by a choice of data domain together with interpretations of some primitive operations. This is formalised in the definition of \emph{semantically-specified algorithm}. The introduced distinction between the control and data will form the basis of the formal definitions introduced thereafter. 

However, while semantically-specified algorithms solve this first issue, they impose constraints that are arguably too strong. Indeed, let us consider the algorithm shown in Figure \ref{gcd1} for computing the greatest common divisor. This is what is usually considered as Euclid's algorithm, a name that is used to refer to different interpretations over different domains. For instance, it already appears in Euclide's \emph{Elements} in distinct forms: in Book VII for integers, and in Book X for magnitudes (e.g. line segments) \cite{EuclidVol1,EuclidVol2,EuclidVol3}. Similarly, Euclid's algorithm can be used to compute the greatest common divisor for polynomials, or for elements of a finite field. 

Under the semantically-specified notion, these would constitute distinct algorithms, since the underlying domains and primitive operations differ. Yet one may still wish to regard them as manifestations of a single algorithmic -- Euclid's algorithm -- which can be instantiated over different structures. I therefore introduce a third notion, that of \emph{logically-specified algorithm}, in which the specification is given in the form of a first-order theory rather than a particular model of that theory. This allows primitive operations to be characterised through their formal properties while remaining independent of a specific semantic realisation.

\begin{figure}
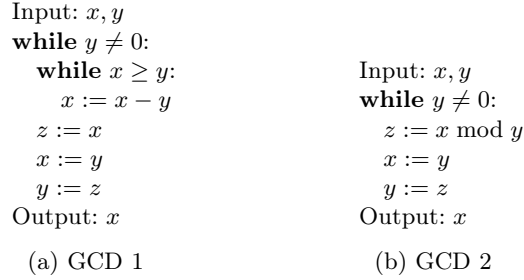

\centering
\subfloat[GCD 1 \label{gcd1}]{
\begin{tabular}{l}
Input: $x,y$\\
\textbf{while} $y \neq 0$: \\
\quad \textbf{while} $x \geq y$: \\
\quad\quad $x := x - y$ \\
\quad $z := x$ \\
\quad $x := y$ \\
\quad $y := z$ \\
Output: $x$
\end{tabular}
}
\hspace{2cm}
\subfloat[GCD 2 \label{gcd2}]{
\begin{tabular}{l}
Input: $x,y$\\
\textbf{while} $y \neq 0$: \\
\quad $z := x \bmod y$ \\
\quad $x := y$ \\
\quad $y := z$ \\
Output: $x$
\end{tabular}
}
\caption{Two presentations of Euclid’s algorithm.}
\label{fig:gcd}
\end{figure}

One key observation is that the proposed approach does not define algorithms as equivalence classes of programs. While complex arguments have been developed to explain how this idea fails \cite{WhenAlgoSame,Dean2016Algorithms}, I propose to approach the question by looking at the two examples of algorithms in Figure \ref{fig:gcd}. Here are two different presentations of Euclid's algorithm. Are those equal? Are those disjoint? The answer to both questions is negative. 

Indeed, programs that do implement the algorithm in Figure \ref{gcd1} are also implementations of the algorithm in Figure \ref{gcd2}, implying that they are not disjoint. However, there are implementations of the algorithm in Figure \ref{gcd2} that are not implementations of the algorithm in Figure \ref{gcd1}: one could compute the operation $\bmod$ using a different technique than iterated subtraction.

Thus, implementation induces a non-trivial inclusion structure between algorithms rather than a partition of programs into equivalence classes. Defining algorithms as equivalence classes of programs would therefore fail to capture such asymmetric implementation relations. The relation induced by implementation is consequently better understood as a preorder than as an equivalence relation. This excludes defining algorithms using a \emph{definition by abstraction}; the current approach could be formulated using a more general approach \cite{seiller-Weyl}. It also points to an essential difference between the current framework and the approach initiated by Yanofsky \cite{Yanofsky}.

\section{Models of computation and programs}\label{sec:models}

I first recall the notions of models of computation, programs, and computations introduced in \cite{seiller-MI1}, which will serve as the underlying framework in this work. The approach is intensional: programs and computations, rather than functions, are taken as primary objects.

\begin{definition}[Model of computation]
A \emph{model of computation} is a monoid action
\(
\alpha:\monoid{I}\acton \Space{X}
\)
where:
\begin{itemize}[nolistsep,noitemsep]
\item $\Space{X}$ is a space, called the \emph{configuration space};
\item $I$ is a set of primitive instructions interpreted by a map\footnote{Here $\Space{X}\to\Space{X}$ usually denotes the set of partial endomorphisms of $\Space{X}$.}
\(
\rho:I\to (\Space{X}\to\Space{X})
\)
extending uniquely\footnote{Note that $\monoid{I}$ does not denote here the free monoid, but the monoid generated by $I$ and the relations induced from the interpretation by $\rho$. In other words, $\monoid{I}$ is isomorphic to the sub-monoid of $\Space{X}\to\Space{X}$ generated by the set $\alpha(I)$.} to the monoid action $\alpha$.
\end{itemize}
\end{definition}

The configuration space can be a set, or a more involved structure (topological space, measure space, topological vector space, etc.); it is expected in the latter case that the group action is compatible with the structure. The elements of $I$ are understood as elementary computational steps, while the monoid $\monoid{I}$ encodes their sequential composition. The action of $\monoid{I}$ on $\Space{X}$ describes the evolution of configurations under composed instructions.

This formulation separates the abstract structure of composition, represented by the monoid $\monoid{I}$, from its operational realisation on configurations through the action on $\Space{X}$.

\begin{example}
The quintessential example is that of the one-tape Turing machine model, represented as the action 
\(\amcTM: \monoid{\InstTM}\acton\SpaceTM\). Here the space of configuration $\SpaceTM$ is defined as the space of $\integerN$-indexed sequences in $\{0,1,\star\}$ which are almost-always equal to $\star$:
\[ \{(s_i)_{i\in\integerN} \mid s_i \in \{0,1,\star\}, \Card\{i\in\integerN \mid s_i\neq \star\}<\infty\}. \] The action $\amcTM$ of the set $\InstTM=\{\instright,\instleft\}\cup\{\instwrite{i}\mid i\in\{0,1,\star\}\}\cup\{\instread{i}\mid i\in\{0,1,\star\}\}$ of atomic instructions is then defined by:
\begin{itemize}[nolistsep,noitemsep]
\item 
\(
\amcTM(\instright): (s_i)_{i\in\integerN}\mapsto(s_{i-1})_{i\in\integerN};
\)
\item 
\(
\amcTM(\instleft): (s_i)_{i\in\integerN}\mapsto(s_{i+1})_{i\in\integerN};
\)
\item 
\(
\amcTM(\instwrite{\star}): (s_i)_{i\in\integerN}\mapsto (t_{i})_{i\in\integerN}
\) where $t_0 = \star$ and $t_i = s_i$ for $i\neq 0$;
\item 
\(
\amcTM(\instread{\star}): (s_i)_{i\in\integerN}\mapsto (s_{i})_{i\in\integerN}
\) 
defined if $s_0=\ast$, undefined otherwise.
\end{itemize}
\end{example}

One can then define the set of programs prescribed by a choice of model of computation.

\begin{definition}[Program]
Let
\(
\alpha:\monoid{I}\acton\Space{X}
\)
be a model of computation. An \emph{$\alpha$-program} is a tuple
\(
P=(S,E,s,t,\ell,\initial{P},\terminal{P})
\)
where:
\begin{itemize}[noitemsep,nolistsep]
\item $(S,E,s,t)$ is a directed graph, with
\(
s,t:E\to S
\)
the source and target maps;
\item
\(
\ell:E\to I
\)
is a labelling of edges by primitive instructions;
\item
\(
\initial{P},\terminal{P}\in S
\)
are distinguished initial and terminal states.
\end{itemize}
By convention, the set $\{ e\in E\mid s(e)=\terminal{P}\}$ is empty.
\end{definition}

The graph structure of the program describes its control flow, with the vertices being the control states. Edge labels specify the elementary instructions performed during transitions between control states.

An $\alpha$-program
\(
P=(S,E,s,t,\ell,\initial{P},\terminal{P})
\)
naturally defines a partial dynamical system on the product space
\(
\Space{X}\times S.
\)
Indeed, each edge $e\in E$ determines a partial transition
\(
(x,s(e))\mapsto (\alpha(\ell(e))(x),t(e)),
\)
combining an evolution of configurations in $\Space{X}$ with a transition between control states. The dynamics generated by these transitions describe the execution of the program.

In particular, a computation of the program from a configuration
\(
x_0\in\Space{X}
\)
is simply an orbit of this partial dynamical system starting from
\(
(x_0,\initial{P}).
\)
Equivalently, a computation is given by a sequence
\(
((x_n,s_n))_{n\in\mathbb{N}}
\)
such that for every $n\in\mathbb{N}$ there exists an edge $e_n\in E$ satisfying
\(
s(e_n)=s_n
\), 
\(
t(e_n)=s_{n+1},
\)
and
\(
x_{n+1}=\alpha(\ell(e_n))(x_n).
\)

The outcome of a computation, when defined, is determined by the configuration component of a state
\(
(x,\terminal{P})
\)
reachable from the initial state. More generally, one may consider observables extracted either from configurations or from complete computation traces.

This framework provides a uniform description of a wide range of computational paradigms. Classical models such as finite automata, Turing machines, or abstract state machines\footnote{A more detailed comparison between those notions and Gurevich's ASM can be found in a paper by Valarcher and Anh-Ton Le \cite{ValarcherLe}.} can be expressed through suitable choices of configuration spaces and instruction sets. More generally, the framework accommodates computational systems arising from algebraic, dynamical, geometric, or non-discrete settings.

Within this framework, programs are concrete computational objects defined relative to a fixed model of computation. The notion of algorithm, introduced below, somehow corresponds to replacing the model of computation used in the definition of program by \emph{abstract data structures}, which we will now introduce.

\section{Abstract data structures}

\subsection{Definitions}

One can already find a notion of "abstract data \emph{type}" in the literature. However, this notion is not operational enough for our purposes, and I will rather work with an abstraction of data \emph{structures}: i.e. a set of \emph{values}, together with a set of allowed operations on those values. Underlying this notion, one can identify a notion of \emph{data representation}. I will therefore define simultaneously two notions: that of data domain and that of data structure, which subsumes a choice of data domain.

\begin{definition}
An \index{abstract data!domain} \emph{abstract data domain} $\datadomain{D}$ is a set $D$. An \index{abstract data!structure} \emph{abstract data structure} $\datastruct{D}$ over an abstract data domain $\datadomain{D}$ is a set of allowed operations $S_\datastruct{D}$ -- called \emph{structural maps} -- in $\sum_{k,k'} \datadomain{D}^k\rightharpoonup \datadomain{D}^{k'}$. Any element $s\in S_\datastruct{D}$ belongs to $\datadomain{D}^k\rightharpoonup\datadomain{D}^{k'}$ for fixed values of $k$ and $k'$: we will write $\dom{s}=k$ and $\im{s}=k'$. The \emph{maximal arity} of $\datastruct{D}$ is defined as\footnote{By definition, we consider that $\max S=\infty$ when the set $S$ is unbounded.}
\(\max \{ \dom{s}, \im{s} \mid s\in S_\datastruct{D} \}.\)
\end{definition}

Let us note that once again, no restrictions are imposed on the cardinality of abstract data domains or on the cardinality of the set of allowed maps. A particularly interesting case is when the latter is defined by induction (\autoref{example:recursive}). Similarly, in order to define a proper, usable, abstract data structure one would require some maps relating the representations of $\datadomain{D}^k$ for various values of $k$. But on principle, no such restrictions will be imposed.

This notion is defined in a way which is quite similar to the definition of abstract model of computation: partial maps are used to allow for reading. For instance, a notion of integers with a test to zero could be represented by the set $\naturalN$ together with maps from $\naturalN$ to $\{0,1\}$ but this would imply the introduction of another data type (booleans), or a particular (ad-hoc) choice of encoding of booleans as specific elements of the set $\naturalN$. Partial functions avoid those by allowing to project on a subset; this can then be used to implement any of these other approaches, and this choice allows for self-contained definitions.

I will now give examples. I note that agreeing on a specific axiomatisation for a given data structure is a complex task, and in fact depends on the expected use. As a consequence, the examples below are proposed axiomatisations but could very well be replaced by others. I am not trying to establish those as \emph{standard} definitions, but rather illustrate how one could manipulate the notion. 

\begin{example}[Abstract boolean\index{abstract data structure!booleans}]\label{data:booleans}
One can define an abstract data structure representing booleans as follows: the underlying data domain is $\datadomain{B}=\{0,1\}$, with the structural maps $\instread{0}$, $\instread{1}$, $\instruction{and}$, $\instruction{and}$, and $\instruction{not}$ defined as:
\begin{align*}
\instread{0}: & n\mapsto n \text{ if and only if $n=0$}\\
\instread{1}: & n\mapsto n \text{ if and only if $n=1$}\\
\instruction{and}: & (m,n)\mapsto m\times n \\
\instruction{or}: & (m,n)\mapsto m+n-m\times n \\
\instruction{not}: & n\mapsto 1-n.
\end{align*}
\end{example}

The following definition captures the standard unary representation of natural numbers, which was used in the early work on computability.

\begin{example}[Natural numbers\index{abstract data structure!natural numbers}]\label{data:natural}
One can define an abstract data structure representing natural numbers as follows: the underlying abstract data domain is $\datadomain{N}=\naturalN$, with the structural maps $\instread{0}$, $\instread{S}$, $\mathtt{succ}$ defined by:
\begin{align*}
\instread{0}: & n\mapsto n \text{ if and only if $n=0$}\\
\instread{S}: & n\mapsto n \text{ if and only if $n\neq 0$}\\
\instruction{succ}: & n\mapsto n+1.
\end{align*}

This data structure can be extended by other operations, such as predecessor, addition and multiplication defined as:
\begin{align*}
\instruction{pred}: & n\mapsto n-1 \text{ if and only if $n\neq 0$}\\
\instruction{add}: & (n,m)\mapsto n+m\\
\instruction{mult}: & (n,m)\mapsto n\times m.
\end{align*}
\end{example}

I leave it to the reader to verify that many other abstract structures commonly used in computability and complexity theory, such as lists, trees or graphs, naturally give rise to abstract data structures in the present sense. Additional examples and constructions can be found in \cite{seiller-HdR}. I nevertheless detail one particularly important example, namely that of recursive functions.

The notion of recursive function is sometimes considered as a model of computation. My understanding is that it is not, because the definition of recursive functions is abstracted from computation; it does not determine computation, it only specifies \emph{what} could be computed. As such, I considered for a while that it could be an example of a programming language. But there are no standard operational semantics associated to recursive functions; one may implement the computation of a recursive function in very different ways. This leads to thinking of the notion as something of a more algorithmic nature, separated from consideration on implementations. It finally occurred to me that it was an example of an abstract data structure. 

\begin{example}[Recursive functions\index{abstract data structure!recursive functions}]
\label{example:recursive}
Primitive recursive and general recursive functions naturally define abstract data structures in the present sense. The underlying data domain is
\(
D=\naturalN,
\)
while the structural maps are given by the corresponding classes of recursive functions.
More precisely, one defines inductively a collection $\mathrm{Rec}$ of functions over $\naturalN$ generated from constant functions, the successor function, and projection functions, and closed under composition and primitive recursion, together with minimisation in the case of general recursive functions.
The resulting collection $\mathrm{Rec}$ therefore provides a canonical example of an abstract data structure generated by closure operations on elementary maps.
\end{example}

There are several consequences of this. First, Turing-completeness can be defined as the possibility to implement this data-structure. Second, it provides a new point of view on algebraic characterisation of complexity classes from Implicit complexity (\icc), such as the Bellantoni and Cook result \cite{bellantonicook}. Those can now be understood as defining abstract data structures that capture time- or space-bounded computations, in the same way linear logic based techniques in \icc capture those same classes through types.

\subsection{Representations}

Obviously, this notion of data structure is independent of a chosen machine model. To compute on data structures, one needs to implement them. To be precise, one would have to account for the complexity of the implementation, which could lead to overhead in time or space when simulating programs from a model of computation into another. However, questions of complexity will be the subject of a later paper and will not be considered here.

\begin{definition}
Suppose given a model of computation $\alpha:\monoid{I}\acton \Space{X}$, and an abstract data domain $\datadomain{D}$ with maximal arity $\mathbf{a}$.
An \index{abstract data!interpretation} interpretation of a data domain $\datadomain{D}$ is a map 
\[ \Delta: \cup_{k\leq \mathbf{a}} D^k\rightarrow \Space{X}.\]
\end{definition}

Let us establish some notations that will be useful below. For each sequence $s\in\{0,1,\star\}^n$, write $\underline{\star} s\underline{\star}$ the configuration $(t_i)_{i\in\integerN}$ defined as $t_i = s_i$ for $i=0,1,\dots,n-1$. Given two sequences $s_1,s_2$, write $s_1\star s_2$ the concatenation of $s_1$ and $s_2$. Lastly, for each natural number $n$, write $\overline{n}^2$ the sequence in $\{0,1\}$ corresponding to its binary representation and $1^n$ the sequence consisting of $n$ occurrences of the symbol $1$.

\begin{example}[Interpretations of Booleans]\label{data:imp:booleans}
Let us consider the Boolean data structure defined above (\autoref{data:booleans}).
Note that it is has maximal arity $2$. 
We can define the interpretation in $\SpaceTM$ as follows:
\[
\begin{array}{rrcl}
\DeltaBool: & a \in \BoolDom  &\mapsto& \underline{\star} a \underline{\star}\\
& (a,b) \in \BoolDom^2 &\mapsto& \underline{\star} a \star b \underline{\star}
\end{array}
\]
\end{example}

\begin{example}[Interpretations of natural numbers]\label{data:imp:integers}
One can define two interpretations of the abstract domain of natural numbers in the Turing machines models $amcTM: \monoid{\InstTM}\acton\SpaceTM$ given as the following:
\begin{align*}
\DeltaNatUn: (a_0,a_1,\dots,a_n) & \mapsto \underline{\star} 1^{a_0}\star 1^{a_1}\star \dots \star 1^{a_n}\underline{\star} \\
\DeltaNatBin: (a_0,a_1,\dots,a_n) & \mapsto \underline{\star} \overline{a_0}^2\star \overline{a_1}^2\star \dots \star \overline{a_n}^2\underline{\star}
\end{align*}
\end{example}

\begin{definition}
Suppose given a model of computation $\alpha:\monoid{I}\acton \Space{X}$, and an abstract data structure $\datastruct{D}=(\datadomain{D},S_D)$. An \index{abstract data!implementation} implementation of $\datastruct{D}$ in $\alpha$ is an interpretation $\Delta$ of the underlying data domain $\datadomain{D}$ together with a map $\delta: S_D\rightarrow \programs{\alpha}$ such that for all $d\in D^k$, $\exists n, \delta(f)^n(\Delta(d))=\Delta(f(d))$.
\end{definition}

The reader is referred to previous work \cite{seiller-HdR} for concrete examples of implementations of the above data structures, and others, in various abstract models of computation.

\subsection{A bit of history}

Recall that the Church-Turing thesis states that the set of  \emph{effectively computable functions\footnote{The original statements considered only functions from natural numbers to natural numbers, although some extensions were considered later on \cite{KleeneRecursiveFunctionals}.}} (which is an informal notion) is equal to the set of functions computable by Turing machines (or equivalently, are computable in lambda-calculus). A forceful argument for the Church-Turing thesis is that the following sets of functions (on natural numbers) are equal: functions computed by Turing machines, and lambda-definable functions.

It is interesting to dive into the research work of the time to understand which results are proven exactly. It turns out that the above equivalence is, from a purely formal point of view, not properly established. I am obviously not putting the result into question, but if one reformulates the exact statements of theorem proven in these papers, one obtains the following:
\begin{itemize}[noitemsep,nolistsep]
\item In his 1937 paper (Computability and $\lambda$-Definability), Turing shows\footnote{Although he writes "No attempt is being made to give a formal proof that this machine has the properties claimed for it. Such a formal proof is not possible unless the ideas of substitution and so forth occurring in the definition of conversion are formally defined, and the best form of such a definition is possibly in terms of machines".} that if $f$ is a lambda-definable function (on Church numerals), then there exists a Turing Machine whose output is the sequence 
\[ \underbrace{11 \dots 1}_{f(0)\text{ symbols}}0\underbrace{11 \dots 1}_{f(1)\text{ symbols}} 0\dots 0 \underbrace{11 \dots 1}_{f(n)\text{ symbols}} 0 \dots \]
\item In the same paper, Turing also shows that if there exists a Turing Machine whose output is the sequence above, then $f$ is general recursive. This is done by means of an encoding of configurations as integers. 
\item In his 1936 paper ($\lambda$-definability and recursiveness), Kleene shows that every general recursive function is $\lambda$-definable using a \emph{shifted version of Church numerals}, in which $0$ is represented as $1$, $2$ represented as $3$, etc.
\end{itemize}

The representation of computable functions by Turing shows that the equivalence considered is based on \emph{unary} representations of integers. Moreover, the choice in the definition of what a computable function is leaves open the question of the representation of partial functions. Indeed, the result applies only to total functions\footnote{This shows that, even though Turing knows that Turing machines compute partial functions, his notion of \emph{computable function} presupposes totality.}. Lastly, one should notice that Kleene uses a different interpretation of (unary) natural numbers.

I will not detail here how the following results can be proved. However, these results can be restated as\footnote{Although binary representations are not considered in these papers, it is easily obtained in a similar manner}:
\begin{itemize}[noitemsep,nolistsep]
\item Lambda-calculus interprets the data structure of partial recursive functions (and can interpret it for both a unary and binary representation of the underlying data domain);
\item Turing machines interpret the data structure of partial recursive functions (and can interpret it for both a unary and binary representation of the underlying data domain);
\item One can encode the elements of $\SpaceTM$ as natural numbers and show that, based on this encoding, each instruction is interpreted as a partial recursive function;
\item One can encode the elements of $\Lambda$ as natural numbers and show that, based on this encoding, that a step of $\beta$-reduction is interpreted as a partial recursive function.
\end{itemize}
Based on the encoding of configurations, one can obtain a simulation of lambda-calculus by Turing machines (and conversely). This simulation is, based on defintions from \cite{seiller-MI1}, a data-constrained program-wise simulation, which means that the above results could be used as an algorithmic Church-Turing thesis, namely that\footnote{Here I restrict the discussion to functions from $\naturalN$ to $\naturalN$ for simplicity.} \emph{the set of effective methods to compute functions is equal to the set of Turing machines (computing functions)}.

Let us note however that this alternative statement does not account for complexity. Indeed, the quantitative aspects of the simulations above imply that time and space constrained computations do not coincide. 

\section{Algorithms}

\subsection{Syntactic structure and glueing}

\begin{definition}
A \index{algorithm!syntactic}syntactic algorithm is a finite labelled graph (with control states) $A=(S,E,s,t,\initial{A},\terminal{A},L,\ell)$ where $S$ is a finite set of \emph{states} containing the initial and terminal states $\initial{A}$ and $\terminal{A}$, $(S,E,s,t)$ is a directed graph, $L$ a finite set of \emph{labels}, and $\ell: E\rightarrow L$ is a labelling function.

The set of all algorithms is written as $\algorithms$.
\end{definition}

The guiding intuition behind this definition is that the graph describes a general control structure, while the labels corresponds to the name of operations used in the algorithm. For the moment, the algorithm is \emph{syntactic} because the labels are just arbitrary names. The notions introduced in the next sections refine the definition by associating labels to more restricted classes of operations. Before going through those definitions, we can already define what it means for a program to \emph{implement} a syntactic algorithm.

\begin{definition}[Glueing]
Let $A=(V,S,\initial{A},\terminal{A},E,s,t,L,\ell)$ be a syntactic algorithm, and an \amc $\alpha:\monoid{I}\acton\Space{X}$. Suppose moreover given a map $\phi: L \rightarrow \programs{\alpha}$. The pre-glueing of $\phi$ along $A$ is the machine $\mathcal{P}(A,\phi)$ defined as the disjoint union $\sum_{e\in E} \phi(\ell(e))$. The \emph{glueing of $\phi$ along $A$} \index{algorithm!glueing} is then defined as the program $\mathcal{G}(A,\phi)$ obtained by identifying for all $v\in V$ the states $\{ \initial{e,\phi(e)}\mid s(e)=v \}\cup \{ \terminal{e,\phi(e)}\mid t(e)=v \}$, and setting $\initial{\mathcal{G}(A,\phi)}=\initial{A}$ and $\terminal{\mathcal{G}(A,\phi)}=\terminal{A}$.
\end{definition}

We illustrate the notion of glueing in \autoref{fig:glueing}. 

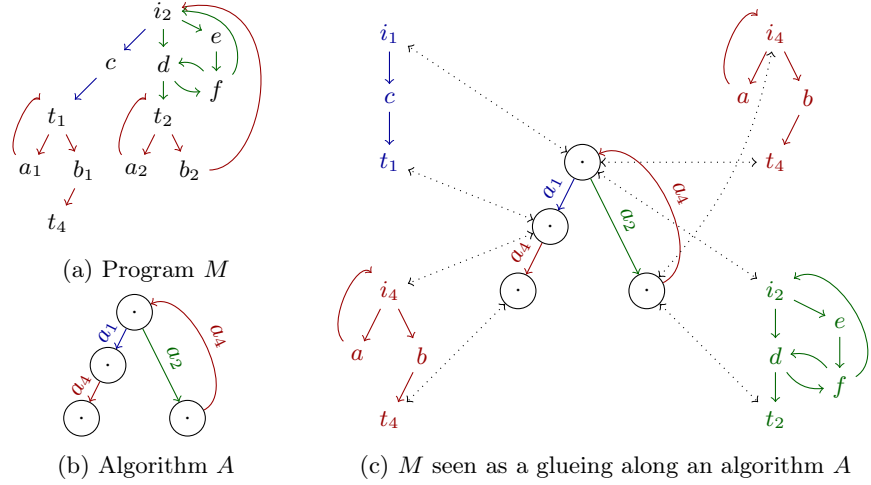
\begin{figure}
\begin{tabular}[b]{c}
\subfloat[Program $M$]{
\begin{tikzpicture}[scale=0.7,baseline=(current bounding box.north)]
\node (inita2) at (3,-1) {\footnotesize{$i_2$}};
\node (terma2) at (3,-3) {\footnotesize{$t_2$}};
\node (internala2d) at (3,-2) {\footnotesize{$d$}};
\node (internala2e) at (4,-1.5) {\footnotesize{$e$}};
\node (internala2f) at (4,-2.5) {\footnotesize{$f$}};
\node (internala42a) at (2.5,-4) {\footnotesize{$a_2$}};
\node (internala42b) at (3.5,-4) {\footnotesize{$b_2$}};
\node (terma1) at (1,-3) {\footnotesize{$t_1$}};
\node (internala1) at (2,-2) {\footnotesize{$c$}};
\node (terma41) at (1,-5) {\footnotesize{$t_4$}};
\node (internala41a) at (0.5,-4) {\footnotesize{$a_1$}};
\node (internala41b) at (1.5,-4) {\footnotesize{$b_1$}};

\draw[->,black!40!red] (terma1) -- (internala41a) {};
\draw[->,black!40!red] (internala41a) .. controls (0,-3.5) and (0.5,-2.5) .. (terma1) {};
\draw[->,black!40!red] (terma1) -- (internala41b) {};
\draw[->,black!40!red] (internala41b) -- (terma41) {};

\draw[->,black!40!blue] (inita2) -- (internala1) {};
\draw[->,black!40!blue] (internala1) -- (terma1) {};

\draw[->,black!40!red] (terma2) -- (internala42a) {};
\draw[->,black!40!red] (internala42a) .. controls (2,-3.5) and (2.5,-2.5) .. (terma2) {};
\draw[->,black!40!red] (terma2) -- (internala42b) {};
\draw[->,black!40!red] (internala42b) .. controls (5,-4) and (5.5,-0.5) .. (inita2) {};

\draw[->,black!60!green] (inita2) -- (internala2d) {};
\draw[->,black!60!green] (internala2d) -- (terma2) {};
\draw[->,black!60!green] (inita2) -- (internala2e) {};
\draw[->,black!60!green] (internala2e) -- (internala2f) {};
\draw[->,black!60!green] (internala2f) to [bend right=30] (internala2d) {};
\draw[->,black!60!green] (internala2d) to [bend right=30] (internala2f) {};
\draw[->,black!60!green] (internala2f) .. controls (4.5,-2) and (4.5,-1) .. (inita2) {};
\end{tikzpicture}
}\\
\subfloat[Algorithm $A$]{
\begin{tikzpicture}[scale=0.7,baseline=(current bounding box.north)]
\node[circle, draw] (begin) at (0,1) {$\cdot$};
\node[circle, draw] (ifl) at (-0.5,0) {$\cdot$};
\node[circle, draw] (ifr) at (1,-1) {$\cdot$};
\node[circle, draw] (end) at (-1,-1) {$\cdot$};
\draw[->,black!40!blue] (begin) -- (ifl) node (a1) [midway,above,sloped] {\footnotesize{$a_1$}};
\draw[->,black!60!green] (begin) -- (ifr) node (a2) [midway,above,sloped] {\footnotesize{$a_2$}};
\draw[->,black!40!red] (ifl) -- (end) node (a41) [midway,above,sloped] {\footnotesize{$a_4$}};
\draw[->,black!40!red] (ifr) .. controls (2,-0.5) and (1,1.5) .. (begin) node (a42) [midway,above,sloped] {\footnotesize{$a_4$}};
\end{tikzpicture}
}
\end{tabular}
\subfloat[$M$ seen as a glueing along an algorithm $A$]{
\begin{tikzpicture}[scale=0.85,baseline=(current bounding box.north)]
\node[circle, draw] (begin) at (0,1) {$\cdot$};
\node[circle, draw] (ifl) at (-0.5,0) {$\cdot$};
\node[circle, draw] (ifr) at (1,-1) {$\cdot$};
\node[circle, draw] (end) at (-1,-1) {$\cdot$};

\draw[->,black!40!blue] (begin) -- (ifl) node (a1) [midway,above,sloped] {\footnotesize{$a_1$}};
\draw[->,black!60!green] (begin) -- (ifr) node (a2) [midway,above,sloped] {\footnotesize{$a_2$}};
\draw[->,black!40!red] (ifl) -- (end) node (a41) [midway,above,sloped] {\footnotesize{$a_4$}};
\draw[->,black!40!red] (ifr) .. controls (2,-0.5) and (1,1.5) .. (begin) node (a42) [midway,above,sloped] {\footnotesize{$a_4$}};

\node[black!40!red] (inita42) at (3,3) {\footnotesize{$i_4$}};
\node[black!40!red] (terma42) at (3,1) {\footnotesize{$t_4$}};
\node[black!40!red] (internala42a) at (2.5,2) {\footnotesize{$a$}};
\node[black!40!red] (internala42b) at (3.5,2) {\footnotesize{$b$}};

\draw[<->,dotted] (ifr) to [bend right=15] (inita42) {};
\draw[<->,dotted] (begin) -- (terma42) {};
\draw[->,black!40!red] (inita42) -- (internala42a) {};
\draw[->,black!40!red] (internala42a) .. controls (2,2.5) and (2.5,3.5) .. (inita42) {};
\draw[->,black!40!red] (inita42) -- (internala42b) {};
\draw[->,black!40!red] (internala42b) -- (terma42) {};

\node[black!40!red] (inita41) at (-3,-1) {\footnotesize{$i_4$}};
\node[black!40!red] (terma41) at (-3,-3) {\footnotesize{$t_4$}};
\node[black!40!red] (internala41a) at (-3.5,-2) {\footnotesize{$a$}};
\node[black!40!red] (internala41b) at (-2.5,-2) {\footnotesize{$b$}};

\draw[<->,dotted] (ifl) -- (inita41) {};
\draw[<->,dotted] (end) -- (terma41) {};
\draw[->,black!40!red] (inita41) -- (internala41a) {};
\draw[->,black!40!red] (internala41a) .. controls (-4,-1.5) and (-3.5,-0.5) .. (inita41) {};
\draw[->,black!40!red] (inita41) -- (internala41b) {};
\draw[->,black!40!red] (internala41b) -- (terma41) {};

\node[black!40!blue] (inita1) at (-3,3) {\footnotesize{$i_1$}};
\node[black!40!blue] (terma1) at (-3,1) {\footnotesize{$t_1$}};
\node[black!40!blue] (internala1) at (-3,2) {\footnotesize{$c$}};

\draw[<->,dotted] (begin) -- (inita1) {};
\draw[<->,dotted] (ifl) -- (terma1) {};
\draw[->,black!40!blue] (inita1) -- (internala1) {};
\draw[->,black!40!blue] (internala1) -- (terma1) {};

\node[black!60!green] (inita2) at (3,-1) {\footnotesize{$i_2$}};
\node[black!60!green] (terma2) at (3,-3) {\footnotesize{$t_2$}};
\node[black!60!green] (internala2d) at (3,-2) {\footnotesize{$d$}};
\node[black!60!green] (internala2e) at (4,-1.5) {\footnotesize{$e$}};
\node[black!60!green] (internala2f) at (4,-2.5) {\footnotesize{$f$}};

\draw[<->,dotted] (begin) -- (inita2) {};
\draw[<->,dotted] (ifr) -- (terma2) {};
\draw[->,black!60!green] (inita2) -- (internala2d) {};
\draw[->,black!60!green] (internala2d) -- (terma2) {};
\draw[->,black!60!green] (inita2) -- (internala2e) {};
\draw[->,black!60!green] (internala2e) -- (internala2f) {};
\draw[->,black!60!green] (internala2f) to [bend right=30] (internala2d) {};
\draw[->,black!60!green] (internala2d) to [bend right=30] (internala2f) {};
\draw[->,black!60!green] (internala2f) to [bend right=90] (inita2) {};
\end{tikzpicture}
}
\caption{Example of glueing}\label{fig:glueing}
\end{figure}

\begin{definition}
A program $G$ \index{algorithm!syntactic!implementation} \emph{implements} the algorithm $A$ when there exists $\phi: L \rightarrow \programs{\alpha}$ such that $G$ is the glueing of $A$ along $\phi$.
\end{definition}

\subsection{Semantically-specified algorithms}

The next step is to impose that labels correspond to specific maps defined as part of data structures. In practice, these data structures are implicitly given when writing down algorithm, e.g. using pseudo-code. 

Let us note that if an algorithm uses two data structures, say integers and booleans, or simply manipulates several elements of a given data structures, the implicit assumption is that several copies of these data structures are simultaneously interpreted within the implementing model of computation. Formally, this is equivalent to considering products of data structures: if $\datastruct{D}=(\datadomain{D},S)$ and $\datastruct{D'}=(\datadomain{D'},S')$ are data structures, then their product $\datastruct{D}\times\datastruct{D'}$ is defined as $(\datadomain{D}\times\datadomain{D'},S\bar{\times} S')$ where: 
\[ S\bar{\times}S' = \{ s\times \identity[\datadomain{D'}] \mid s\in S \}\cup  \{ \identity[\datadomain{D}]\times s' \mid s'\in S' \}. \]
As a consequence, one can define algorithms w.r.t. a single data structure.

\begin{definition}[Semantically-specified algorithm]
Let $\datastruct{D}=(\datadomain{D},S)$ be a data structure. A \emph{semantically-specified algorithm} with respect to $\datastruct{D}$ is defined as a tuple $A=(V,\initial{A},\terminal{A},E,s,t,L,\ell,\Interpret{\cdot})$ where:
\begin{itemize}[noitemsep,nolistsep]
\item $(V,\initial{A},\terminal{A},E,s,t,L,\ell)$ is a syntactical algorithm;
\item $\Interpret{\cdot}: L \rightarrow S$ maps the labels to structural maps.
\end{itemize}
The set of all specified algorithms w.r.t. $\datastruct{D}$ is written $\algorithms[\datastruct{D}]$.
\end{definition}

\begin{definition}[Coherent labelings]
Suppose given a specified algorithm $A=(V,S,\initial{A},\terminal{A},E,s,t,L,\ell,\Interpret{\cdot})$ for a data structure $\datastruct{D}=(\datadomain{D},S)$, and a model of computation $\alpha:\monoid{I}\acton\Space{X}$. A map $\phi: L \rightarrow \programs{\alpha}$ is \emph{coherent} with $\Interpret{\cdot}$ with respect to an interpretation $\Delta$ of $\datadomain{D}$ when for all $o\in L$, $\phi(o)$ is an implementation of $\Interpret{o}$ w.r.t. $\Delta$.
\end{definition}

\begin{definition}
A program $G$ \emph{implements} the specified algorithm $A$ when there exists $\phi: L \rightarrow \programs{\alpha}$ coherent with $\Interpret{\cdot}$ w.t. $G$ is the glueing of $A$ along $\phi$.
\end{definition}

This notion of semantically-specified algorithm could probably be stronger than what one would expect. One important remark here is that the notion of data structure is set-theoretic and requires one to fix a concrete domain. As a consequence, the notion does not allow to identify Euclidean division on integers and Euclidean division on polynomials because the underlying abstract data domain is different. The following proposed definition proposes a solution to this problem.

\subsection{Logically specified data structures}

\begin{definition}
A \emph{logical data structure} $\mathcal{D}$ is defined as first-order logical theory over a first-order language $(\mathrm{Var},\mathrm{Fun},\mathrm{Rel})$. 

An abstract data structure $\datastruct{D}$ is \emph{a model of the logical data structure $\mathcal{D}$} if it is a model of $\mathcal{D}$, i.e. if $\datastruct{D}\vDash \mathcal{D}$. 
\end{definition}

\begin{definition}[Logically-specified algorithm]
Let $\mathcal{D}$ be a logical data structure over a first-order language $(\mathrm{Var},\mathrm{Fun},\mathrm{Rel})$. A \emph{logically-specified algorithm} with respect to $\mathcal{D}$ is a tuple $A=(V,\initial{A},\terminal{A},E,s,t,L,\ell,\Interpret{\cdot})$ where:
\begin{itemize}[noitemsep,nolistsep]
\item $(V,\initial{A},\terminal{A},E,s,t,L,\ell)$ is a syntactical algorithm;
\item $\Interpret{\cdot}: L \rightarrow \mathrm{Fun}\cup\mathrm{Rel}$ maps the labels to structural maps.
\end{itemize}
The set of all specified algorithms w.r.t. $\mathcal{D}$ is written $\algorithms[\mathcal{D}]$.
\end{definition}

Given a logically-specified algorithm $A$ w.r.t. $\mathcal{D}$ and a model $\datastruct{D}$ of $\mathcal{D}$, one can deduce the unique induced semantically-specified algorithm $A_{\mathcal{D}\vDash \datastruct{D}}$. Implementing $A$ then boils down to implementing $A_{\mathcal{D}\vDash \datastruct{D}}$ for a model $\datastruct{D}$ of $\mathcal{D}$.

\begin{definition}
A program $P$ implements the logically specified algorithm $A$ if there exists an abstract data structure $\datastruct{D}$ s.t. $P$ implements the induced specified algorithm $A_{\mathcal{D}\vDash \datastruct{D}}$.
\end{definition}

Note that Euclidean division on integers and on arbitrary Euclidean ring are but different models of the same logical data structure. The corresponding notion of \emph{logically specified algorithm}, intermediate between the syntactical and specified notions considered above, then allows to identify the Euclidean division algorithms defined on different Euclidean rings.

The reader familiar with Gurevich's \emph{abstract state machines} \cite{ASMgurevich} will probably be curious about potential connections with the notion just developed since both are based on first order structures. While there may be a formal relation between the approaches, I note that Gurevich uses the first-order structure to define the states (i.e. the space underlying the \amc) while it is here used to describe instructions (or rather, more precisely, programs) that can be performed. This seems to be a fundamental difference. One consequence of this is that Gurevich modifies the values of the interpretations of functions and relation symbols, while here the interpretation is fixed once and for all.

\subsection{Properties}

The definition of \emph{glueing} can be adapted to define glueing of algorithms (instead of programs) along another algorithm. This leads to the definition of a preorder on the set of algorithms.

\begin{definition}[Glueing of algorithms along a labelled graph]
Suppose given a syntactic algorithm $A=(V,S,\initial{A},\terminal{A},E,s,t,L,\ell)$, and an \amc $\alpha:\monoid{I}\acton\Space{X}$. Suppose moreover given a map $\phi: L \rightarrow \algorithms$. The pre-glueing of $\phi$ along $A$ is defined as the disjoint union $\mathcal{P}(A,\phi)=\sum_{e\in E} \phi(\ell(e))$. The \emph{glueing of $\phi$ along $A$} \index{algorithm!glueing (algorithm)} is then defined as the algorithm $\mathcal{G}(A,\phi)$ obtained by identifying for all $v\in V$ the vertices $\{ \initial{e,\phi(e)}\mid s(e)=v \}\cup \{ \terminal{e,\phi(e)}\mid t(e)=v \}$, and setting $\initial{\mathcal{G}(A,\phi)}=\initial{A}$ and $\terminal{\mathcal{G}(A,\phi)}=\terminal{A}$.

If $\phi$ maps $L$ into $\algorithms[D]$ (for $D$ either an abstract data structure or a logical data structure), the resulting algorithm belongs to $\algorithms[D]$.
\end{definition}

\begin{proposition}\label{prop:algoorder}
If $B$ is obtained as the glueing of $A$ along $\phi$, and $P$ is a program implementing $B$, then $P$ also implements $A$.
\end{proposition}

\begin{proof}
Here is a sketch of the proof, avoiding painful details. If $P$ implements $B$, then $P$ is a glueing of $B$ along some $\psi:L^B\rightarrow\programs{\alpha}$. Moreover, $B$ is a glueing of $A$ along some $\phi:L^{A}\rightarrow \algorithms$. Now, from both these maps one can deduce a $\theta: L^A\rightarrow\programs{\alpha}$. It is not difficult to check that $P$ is the glueing of $A$ along this map $\theta$.
\end{proof}

The notion of algorithm obtained does not correspond to an equivalence relation on programs. In particular, one program implements many different algorithms. In some way, the notion is almost topological, and one could expect some separation axiom to be satisfied (probably the $T_0$ axiom: that if $P$ and $Q$ are different programs, there exists an algorithm $A$ such that one of $P$ and $Q$ implement $A$ but not the other).

Among the future research in that direction, one direction seems of particular importance: the definition of a notion of \emph{distance}. A distance between algorithms, together with a notion of \emph{implementation/glueing} up to some error $\epsilon$. This could be used to talk about convergence of programs toward an algorithm for instance.

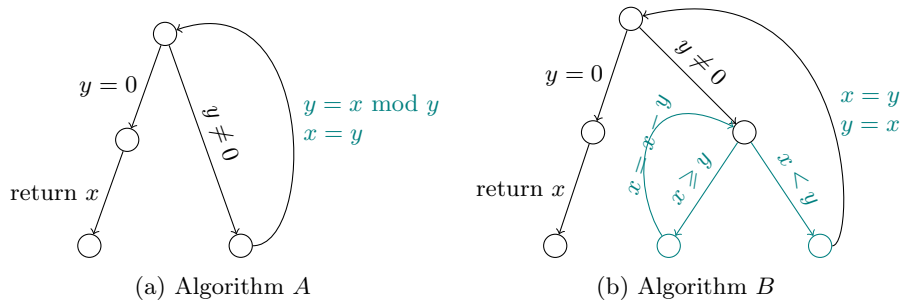
\begin{figure}
\subfloat[Algorithm $A$]{
\begin{tikzpicture}[y=1.4cm]
\node[circle, draw] (begin) at (0,0) {~};
\node[circle, draw] (ifl) at (-0.5,-1) {~};
\node[circle, draw] (ifr) at (1,-2) {~};
\node[circle, draw] (end) at (-1,-2) {~};

\draw[->] (begin) -- (ifl) node [midway,left] {\footnotesize{$y=0$}};
\draw[->] (begin) -- (ifr) node [midway,above,sloped] {\footnotesize{$y\neq0$}};
\draw[->] (ifl) -- (end) node [midway,left] {\footnotesize{$\textrm{return }x$}};
\draw[->] (ifr) .. controls (2,-2) and (2,0.5) .. (begin) node [midway,right] {\footnotesize{\begin{tabular}{l}$\textcolor{teal}{y=x\mathrm{~mod~}y}$\\$\textcolor{teal}{x=y}$\end{tabular}}};
\end{tikzpicture}
}
\subfloat[Algorithm $B$]{
\begin{tikzpicture}[y=1.5cm]
\node[circle, draw] (begin) at (0,0) {~};
\node[circle, draw] (ifl) at (-0.5,-1) {~};
\node[circle, draw] (ifr) at (1.5,-1) {~};
\node[circle, draw] (end) at (-1,-2) {~};
\node[circle, draw, teal] (ifrifl) at (0.5,-2) {~};
\node[circle, draw, teal] (ifrifr) at (2.5,-2) {~};

\draw[->] (begin) -- (ifl) node [midway,left] {\footnotesize{$y=0$}};
\draw[->] (begin) -- (ifr) node [midway,above,sloped] {\footnotesize{$y\neq0$}};
\draw[->] (ifl) -- (end) node [midway,left] {\footnotesize{$\textrm{return }x$}};
\draw[->,teal] (ifr) -- (ifrifl) node [midway,above,sloped] {\footnotesize{$x\geqslant y$}};
\draw[->,teal] (ifr) -- (ifrifr) node [midway,sloped,above] {\footnotesize{$x<y$}};
\draw[->] (ifrifr) .. controls (3,-2) and (3,0.5) .. (begin) node [midway,right] {\footnotesize{\begin{tabular}{l}$\textcolor{teal}{x=y}$\\$\textcolor{teal}{y=x}$\end{tabular}}};
\draw[->,teal] (ifrifl) .. controls (0,-1.5) and (0,-0.5) .. (ifr) node [sloped,midway] {\footnotesize{\begin{tabular}{l}$x=x-y$\end{tabular}}};
\end{tikzpicture}
}
\caption{Two specified algorithms}\label{fig:specalgorithms}
\end{figure}

\section{Examples}

\paragraph{Variants of the \textsc{gcd} algorithm.}

As a first example, I simply give in Figure \ref{fig:specalgorithms} two representations of the \textsc{gcd} algorithms shown in Figure \ref{fig:gcd}. These algorithms can be considered as \emph{logically specified}, where the data structure is defined as a model of a first order theory of Euclidean rings. Note that these provide an example of the preorder just defined: the algorithm $B$ is obtained as a glueing of some specific algorithm for computing the remainder along the algorithm $A$.

\paragraph{The merge sort algorithm.}\label{sec:mergesort}

The merge sort algorithm, and more generally the recursive algorithms, are examples of algorithms which are simultaneously defined with the data structure. I.e. the merge sort algorithm is defined as a specified algorithm with respect to a data structure that contains a sorting algorithm. 

More precisely, suppose one wants to define the merge sort algorithm. The first thing is to define the data structure. One will require here a structure that allows for working with several lists, and compare elements (say integers). Defining a recursive algorithm then corresponds to the following: one adds to the data structure the corresponding structural map $\instruction{sort}:\DeltaList\rightarrow\DeltaList$. To ease the argument, suppose that there exists a structural map $\instruction{split}:\DeltaList\rightarrow\DeltaList\times\DeltaList$ defined as $s_0,s_1,\dots,s_k\mapsto (s_0,s_2,s_4,\dots)\times(s_1,s_3,s_5,\dots)$. It should be clear that if such a map is not given as part of the structure, one can define a corresponding algorithme using more elementary operations on list. One can then define the sorting algorithm shown in \autoref{fig:mergesort}. Note that in the figures, the initial state is represented by an incoming edge and the terminal state by a double circle. 

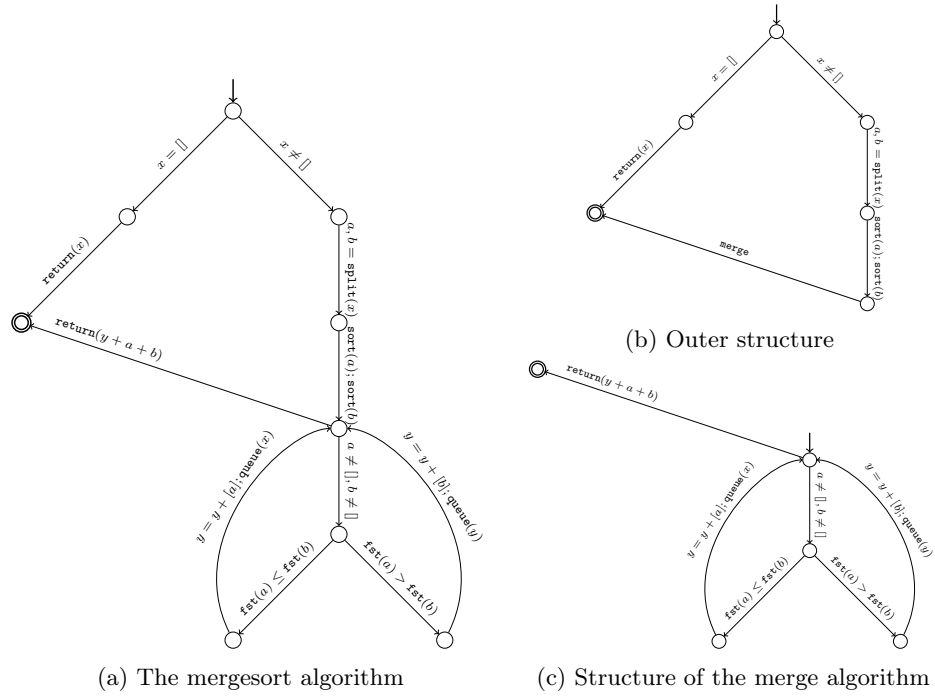
\begin{figure}[t]
\centering

\subfloat[The mergesort algorithm\label{fig:mergesort}]{
\scalebox{0.7}{
\begin{tikzpicture}[x=2cm,y=2cm]
\node[circle, draw] (init) at (0,0) {~};
\node[circle, draw] (left) at (-1,-1) {~};
\node[circle, draw] (right) at (1,-1) {~};
\node[circle, draw] (split) at (1,-2) {~};
\node[circle, draw] (sorted) at (1,-3) {~};
\node[circle, draw] (nonempty) at (1,-4) {~};
\node[circle, draw] (combinel) at (0,-5) {~};
\node[circle, draw] (combiner) at (2,-5) {~};
\node[circle, draw, thick, double] (end) at (-2,-2) {~};

\draw[->,thick] (0,0.3) -- (init);

\draw[->] (init) -- (left)
node [midway,above,sloped] {\scriptsize{$x=[]$}};
\draw[->] (init) -- (right)
node [midway,above,sloped] {\scriptsize{$x\neq []$}};
\draw[->] (right) -- (split)
node [midway,above,sloped] {\scriptsize{$a,b=\instruction{split}(x)$}};
\draw[->] (split) -- (sorted)
node [midway,above,sloped] {\scriptsize{$\instruction{sort}(a);\instruction{sort}(b)$}};
\draw[->] (sorted) -- (nonempty)
node [midway,above,sloped] {\scriptsize{$a\neq[], b\neq []$}};
\draw[->] (nonempty) -- (combinel)
node [midway,sloped,above] {\scriptsize{$\instruction{fst}(a)\leq \instruction{fst}(b)$}};
\draw[->] (nonempty) -- (combiner)
node [midway,sloped,above] {\scriptsize{$\instruction{fst}(a)> \instruction{fst}(b)$}};
\draw[->] (combinel)
.. controls (-0.5,-4) and (0.5,-3) ..
(sorted)
node [midway,sloped,above]
{\scriptsize{$y=y+[a];\instruction{queue}(x)$}};
\draw[->] (combiner)
.. controls (2.5,-4) and (1.5,-3) ..
(sorted)
node [midway,sloped,above]
{\scriptsize{$y=y+[b];\instruction{queue}(y)$}};
\draw[->] (left) -- (end)
node [midway,above,sloped]
{\scriptsize{$\instruction{return}(x)$}};
\draw[->] (sorted) -- (end)
node [near end,above,sloped]
{\scriptsize{$\instruction{return}(y+a+b)$}};
\end{tikzpicture}
}
}
\begin{tabular}[b]{c}
\subfloat[Outer structure \label{fig:recmergesort}]{
\scalebox{0.6}{
\begin{tikzpicture}[x=2cm,y=2cm]
\node[circle, draw] (init) at (0,0) {~};
\node[circle, draw] (left) at (-1,-1) {~};
\node[circle, draw] (right) at (1,-1) {~};
\node[circle, draw] (split) at (1,-2) {~};
\node[circle, draw] (sorted) at (1,-3) {~};
\node[circle, draw, thick, double] (end) at (-2,-2) {~};

\draw[->,thick] (0,0.3) -- (init);

\draw[->] (init) -- (left)
node [midway,above,sloped] {\scriptsize{$x=[]$}};
\draw[->] (init) -- (right)
node [midway,above,sloped] {\scriptsize{$x\neq []$}};
\draw[->] (right) -- (split)
node [midway,above,sloped] {\scriptsize{$a,b=\instruction{split}(x)$}};
\draw[->] (split) -- (sorted)
node [midway,above,sloped] {\scriptsize{$\instruction{sort}(a);\instruction{sort}(b)$}};
\draw[->] (sorted) -- (end)
node [midway,above,sloped] {\scriptsize{$\instruction{merge}$}};
\draw[->] (left) -- (end)
node [midway,above,sloped]
{\scriptsize{$\instruction{return}(x)$}};
\end{tikzpicture}
}
}
\\
\subfloat[Structure of the merge algorithm\label{fig:recmerge}]{
\scalebox{0.6}{
\begin{tikzpicture}[x=2cm,y=2cm]
\node[circle, draw] (sorted) at (1,-3) {~};
\node[circle, draw] (nonempty) at (1,-4) {~};
\node[circle, draw] (combinel) at (0,-5) {~};
\node[circle, draw] (combiner) at (2,-5) {~};
\node[circle, draw, thick, double] (end) at (-2,-2) {~};

\draw[->,thick] (1,-2.7) -- (sorted);

\draw[->] (sorted) -- (nonempty)
node [midway,above,sloped] {\scriptsize{$a\neq[], b\neq []$}};
\draw[->] (nonempty) -- (combinel)
node [midway,sloped,above]
{\scriptsize{$\instruction{fst}(a)\leq \instruction{fst}(b)$}};
\draw[->] (nonempty) -- (combiner)
node [midway,sloped,above]
{\scriptsize{$\instruction{fst}(a)> \instruction{fst}(b)$}};
\draw[->] (combinel)
.. controls (-0.5,-4) and (0.5,-3) ..
(sorted)
node [midway,sloped,above]
{\scriptsize{$y=y+[a];\instruction{queue}(x)$}};
\draw[->] (combiner)
.. controls (2.5,-4) and (1.5,-3) ..
(sorted)
node [midway,sloped,above]
{\scriptsize{$y=y+[b];\instruction{queue}(y)$}};
\draw[->] (sorted) -- (end)
node [near end,above,sloped]
{\scriptsize{$\instruction{return}(y+a+b)$}};
\end{tikzpicture}
}
}
\end{tabular}
\caption{The mergesort algorithm.}
\label{fig:mergesort-decomposition}
\end{figure}

Now if the $\instruction{sort}$ label is interpreted by a program that sorts a list (whichever method is used), a program implementing this algorithm is itself a program that sorts a list. The recursive definition consists in defining the overall algorithm $R$ obtained by glueing $A$ along the label $\instruction{sort}$ within $A$. This is computing a limit: one starts with $A$, and then obtains $A[\instruction{sort}\leftarrow A]$, in which one then glues $A$ along the label $\instruction{sort}$ to obtain $A[\instruction{sort}\leftarrow A][\instruction{sort}\leftarrow A]$, etc.

The recursive program $R$ is then defined formally as the limit of this process, which can be represented as an infinite graph that no longer contains the label $\instruction{sort}$. If one tries to express in more details the recursive structure, one obtains that the obtained algorithm is defined as an implementation of the algorithm shown in \autoref{fig:recmergesort} which put forth the recursive structure, which is close to the recursive definition considered by Moschovakis \cite{MoschovakisAlgo}:
\[
\mathrm{sort}(x)=\left\{
\begin{array}{ll}
[] &\text{if $x=[]$}\\
\mathrm{merge}(\mathrm{sort}(\mathrm{split}(x)_1),\mathrm{sort}(\mathrm{split}(x)_2)) & \text{otherwise}.
\end{array}\right.
\]
Here the $\mathrm{merge}$ algorithm, shown in \autoref{fig:recmerge}, can also be given a recursive definition:
\[
\mathrm{merge}(a,b)=\left\{
\begin{array}{ll}
b &\text{if $a=[]$}\\
a &\text{if $b=[]$}\\
\mathrm{fst}(a)+\mathrm{merge}(\mathrm{queue}(a),b) & \text{if $\mathrm{fst}(a)\leq \mathrm{fst}(b)$}\\
\mathrm{fst}(b)+\mathrm{merge}(a,\mathrm{queue}(b)) & \text{otherwise}\\
\end{array}\right.
\]

Finally, let us note that in the proposed algorithm, no order is enforced on the two subroutines $\instruction{sort}(a)$ and $\instruction{sort}(b)$: the label simply imposes that both operation are performed in this part of the algorithm. As a consequence, an implementation may choose to sort $a$ before $b$, or $b$ before $a$, or even sort both in parallel. Variants of this algorithm in which an order is imposed can be considered by replacing the arrow labelled with both sort operations by two consecutive arrows, one sorting $a$ (for instance), the other sorting $b$. Another variant in which both operations are performed in parallel may be represented by a single arrow labelled with $\instruction{sort}(a)\mid \instruction{sort}(b)$ where $\mid$ indicates explicitly the parallel execution. These algorithms are refinements of the original one, and any implementation of those is an implementation of the algorithm in \autoref{fig:mergesort}. However they are more specific, and any implementation of the original algorithm may not implement one of those, depending on how the subroutine are ordered.

\section{Algorithms as succinct descriptions}

The previous sections introduced a formal definition of algorithms complementing the definition of programs recalled in Section \ref{sec:models}. At a structural level, the two notions are deliberately very close: both are described by finite control graphs labelled by computational operations. This similarity reflects the fact that algorithms and programs both describe computational processes, albeit at different levels of abstraction.

However, this proximity also raises a conceptual issue. If algorithms and programs are represented by essentially similar mathematical objects, what distinguishes an algorithm from a program? I argue that the distinction cannot rest solely on the formal structure itself, but must also involve the role played by these objects. In particular, an algorithm should provide a succinct and intelligible description of the programs implementing it.

This question becomes especially relevant in the context of neural networks. Indeed, the complete sequence of matrices defining a neural network can itself be described abstractly using an adequate abstract data structure. One might therefore argue that a neural network already constitutes an algorithm, with its concrete execution corresponding merely to an implementation in which each operation is realised directly as a primitive computational step. While this would be coherent with the third usage of the term \emph{algorithm} mentioned in the introduction, it feels incompatible with the intuitive notion of algorithm developed in mathematics and computer science.

I will therefore argue instead that such a fully explicit description should not, in general, be regarded as algorithmically informative. Although it formally specifies a computational process, it somehow fails to do so at a sufficiently abstract level. Algorithms are not merely abstract specifications: they are also intended to provide intelligible and communicable descriptions of computational processes.

This suggests refining the relationship between algorithms and programs by introducing a notion of succinctness. Rather than modifying the definition of algorithm itself, the idea is to distinguish programs that admit genuinely compressed algorithmic descriptions from those that do not.

\begin{definition}
Fix a function
\(
f:\mathbb{N}\to\mathbb{N}
\)
such that
\(
f(n)=o(n).
\)
Given a program $P$, we say that $P$ \emph{admits an $f$-succinct algorithm} if there exists an algorithm $A$ implemented by $P$ satisfying
\(
\mathrm{size}(A)\leq f(\mathrm{size}(P)).
\)
\end{definition}

The precise choice of the function $f$ determines the degree of compression required from the algorithmic description relative to the implementation. I will not attempt here to determine which asymptotic regimes best reflect practical notions of succinctness. Such questions would require a more detailed study of concrete implementations, for example comparing algorithms with compiled executable code rather than source-level descriptions. The arguments developed below, however, only require the assumption that
\(
f(n)=o(n).
\)

This condition admits a natural interpretation in terms of Kolmogorov complexity. Indeed, if a program $P$ admits an $f$-succinct algorithm $A$, together with a fixed implementation procedure transforming $A$ into $P$, then $A$ is a witness that $P$ has low Kolmogorov complexity.

In the context of neural networks, admitting a succinct algorithm may be interpreted as a strong intensional form of explainability. Indeed, the existence of a compressed algorithmic description means that the network is not merely a large collection of parameters, but possesses an underlying computational structure describable at a significantly higher level of abstraction.

The following consequence is immediate from standard counting arguments in Kolmogorov complexity theory.

\begin{proposition}
The proportion of programs of size $n$ admitting $f$-succinct algorithms tends to zero as $n\to\infty$.
\end{proposition}

In this sense, most sufficiently large programs do not possess meaningful algorithmic descriptions. Equivalently, most programs are not compressible into significantly smaller abstract computational structures. From the present perspective, they therefore fail to exhibit the kind of higher-level organisation usually associated with algorithms. In layman's term, this theorem states that most programs do not admit $f$-succinct algorithms.

This observation is particularly relevant in the context of large neural networks. While such systems may implement highly structured behaviours, the proposition suggests that generic parameter configurations should not be expected to admit succinct algorithmic descriptions. Possessing an algorithm, in the sense discussed above, therefore corresponds to a strong intensional form of explainability: it means that the computational behaviour of the network can be described through a substantially more compact abstract structure.

\section{Conclusion}

The distinction introduced in this work between algorithms, programs, and implementations opens several directions for future research. Since implementations induce a preorder relating abstract algorithms to concrete programs, one may investigate quantitative notions comparing different implementations or measuring the extent to which a program admits an underlying algorithmic structure. This raises the possibility of defining notions analogous to distances, approximations, or convergence between programs and algorithms, potentially allowing one to formalise ideas such as a program ``converging'' toward an algorithmic description, which could shed some light on learning techniques.

In particular, from the present perspective, an important question is whether such systems possess recoverable intensional structures that may legitimately be regarded as succinct algorithms. If such structure exists, then it should be describable from the training data, independently of the particular underlying model of computation. In a recent collaboration with Jarvis, Gastaldi, and Terilla \cite{seiller-GeoNucleus,seiller-LogicNucleus}, we explore geometric and logical structures arising from data. Reframed in the context of the current discussion, this corresponds to the extraction of an abstract data structure from the training corpus directly, and points to the possibility of mathematically recovering some succinct algorithmic description of the programs obtained through training.

The formal definition of algorithms and implementation could also provide the basis for the development of formal verification and certification techniques at the algorithmic level itself. Rather than verifying individual implementations independently, one could aim to establish properties and correctness results directly for algorithms and then transfer these results systematically to concrete programs through the implementation relation. This would provide a mathematically robust framework for implementation-independent verification and certification, allowing proofs carried out at the abstract algorithmic level to apply uniformly across heterogeneous computational models and implementations.

Finally, the framework also suggests a new approach to computational complexity. Classical complexity theory deliberately abstracts away from concrete implementation details and machine architectures. By contrast, the present approach explicitly relates algorithms, representations, and implementations, making it possible to define complexity measures sensitive to the actual structure of computations. This raises the prospect of developing exact and architecture-aware complexity theories incorporating features such as memory hierarchies, cache sizes, branch prediction, or parallel execution mechanisms, while remaining grounded in abstract algorithmic descriptions.

\bibliographystyle{splncs04}
\bibliography{thomas}

\begin{thebibliography}{10}
\providecommand{\url}[1]{\texttt{#1}}
\providecommand{\urlprefix}{URL }
\providecommand{\doi}[1]{https://doi.org/#1}

\bibitem{Airoldi}
Airoldi, M.: Machine {Habitus}: {Toward} a {Sociology} of {Algorithms}. Polity
  Press (2022), \url{https://books.google.fr/books?id=sIVdzgEACAAJ}

\bibitem{ValarcherLe}
Anh-Ton~Le, H.L., Valarcher, P.: {Completeness of Seiller's Abstract Machine}
  (2025), \url{https://hal.u-pec.fr/hal-05137612}, preprint

\bibitem{bellantonicook}
Bellantoni, S., Cook, S.: A new recursion-theoretic characterization of the
  polytime functions. Computational Complexity  \textbf{2} (1992),
  \url{https://doi.org/10.1007/BF01201998}

\bibitem{BlassGurevich}
Blass, A., Gurevich, Y.: Algorithms: A quest for absolute definitions. Bulletin
  of the European Association for Theoretical Computer Science  (2003)

\bibitem{WhenAlgoSame}
Blass, A., Dershowitz, N., Gurevich, Y.: When are two algorithms the same? CoRR
   \textbf{abs/0811.0811} (2008), \url{http://arxiv.org/abs/0811.0811}

\bibitem{ChurchAlgo}
Church, A.: An unsolvable problem of elementary number theory. American Journal
  of Mathematics  \textbf{58}(2),  345--363 (1936),
  \url{http://www.jstor.org/stable/2371045}

\bibitem{Dean2016Algorithms}
Dean, W.: Algorithms and the mathematical foundations of computer science. In:
  Horsten, L., Welch, P. (eds.) Gödel's Disjunction: The scope and limits of
  mathematical knowledge. Oxford University Press (08 2016).
  \doi{10.1093/acprof:oso/9780198759591.003.0002}

\bibitem{seiller-LogicNucleus}
Gastaldi, J.L., Jarvis, S., Seiller, T., Terilla, J.: Linear realisability
  structures in enriched adjunctions (2026), submitted

\bibitem{seiller-GeoNucleus}
Gastaldi, J.L., Jarvis, S., Seiller, T., Terilla, J.: Projective metric
  geometry of tropical nuclei: gap matrices, event loci, and order chambers
  (2026), submitted

\bibitem{ASMgurevich}
Gurevich, Y.: Sequential abstract state machines capture sequential algorithms.
  ACM Transactions on Computational Logic  \textbf{1},  77--111 (2000)

\bibitem{GurevichAlgo}
Gurevich, Y.: What Is an Algorithm?, pp. 31--42. Springer Berlin Heidelberg
  (2012). \doi{10.1007/978-3-642-27660-6_3}

\bibitem{EuclidVol1}
Heath, T.: The {Thirteen} {Books} of {Euclid}'s {Elements}. No. vol. 1,
  Cambridge University Press (1956)

\bibitem{EuclidVol2}
Heath, T.: The {Thirteen} {Books} of {Euclid}'s {Elements}. No. vol. 2,
  Cambridge University Press (1956)

\bibitem{EuclidVol3}
Heath, T.: The {Thirteen} {Books} of {Euclid}'s {Elements}. No. vol. 3,
  Cambridge University Press (1956)

\bibitem{seiller-Weyl}
Joinet, J.B., Seiller, T.: From abstraction and indiscernibility to
  classification and types: revisiting hermann weyl’s theory of ideal
  elements. Kagaku tetsugaku  \textbf{53}(2),  65--93 (2021).
  \doi{10.4216/jpssj.53.2_65}

\bibitem{KleeneRecursiveFunctionals}
Kleene, S.C.: Recursive functionals and quantifiers of finite types i.
  Transactions of the American Mathematical Society  \textbf{91}(1),  1--52
  (1959), \url{http://www.jstor.org/stable/1993145}

\bibitem{KolmogorovUspensky}
Kolmogorov, A.N., Uspenskii, V.A.: On the definition of an algorithm. Uspekhi
  Mat. Nauk  \textbf{13},  3--28 (1958)

\bibitem{LamasseAlgoriste}
Lamassé, S.: Relationships between french “practical arithmetics” and
  teaching? In: Scientific Sources and Teaching Contexts Throughout History:
  Problems and Perspectives, pp. 125--153. Springer (2013)

\bibitem{MarkovAlgo}
Markov, A.A.: The theory of algorithms, Trudy Mat. Inst. Steklov., vol.~42.
  Acad. Sci. USSR (1954)

\bibitem{MoschovakisAlgo}
Moschovakis, Y.: What is an algorithm? In: Mathematics Unlimited --- 2001 and
  beyond (2001)

\bibitem{MoschovakisFoundations}
Moschovakis, Y.N.: On founding the theory of algorithms. In: Dales, H.G.,
  Oliveri, G. (eds.) Truth in Mathematics, pp. 71--104. Oxford University
  Press, Usa (1998)

\bibitem{GoA}
Naibo, A., Petrolo, M., Seiller, T.: Goa: The geometry of algorithms. The
  Reasoner  \textbf{17}(4) (Jul 2023),
  \url{https://riviste.unimi.it/index.php/thereasoner/article/view/24136}

\bibitem{seiller-HdR}
Seiller, T.: Mathematical informatics (2024),
  \url{https://theses.hal.science/tel-04616661}, habilitation thesis

\bibitem{seiller-MI1}
Seiller, T.: Mathematical informatics: Models of computation (2026),
  \url{https://hal.archives-ouvertes.fr/hal-05587108}, submitted

\bibitem{Yanofsky}
Yanofsky, N.S.: Towards a definition of an algorithm. Journal of Logic and
  Computation  \textbf{21}(2),  253--286 (2011). \doi{10.1093/logcom/exq016}

\end{thebibliography}

\end{document}